\documentclass[12pt]{amsart}
\usepackage{latexsym,fancyhdr,amssymb,color,amsmath,amsthm,graphicx,listings,comment}
\usepackage[absolute,overlay]{textpos}
\usepackage[section]{placeins}
\newtheorem{thm}{Theorem} 
\newtheorem{lemma}{Lemma} 
\newtheorem{coro}{Corollary}
\let\paragraph\subsubsection

\title{Block Jacobi matrices, 
       Barycentric limits and 
       Manifolds}

\author{Oliver Knill}
\date{January 14, 2026, Dedicated to Barry Simon on the occasion of his 80th birthday}
\address{Department of Mathematics \\ Harvard University \\ Cambridge, MA, 02138 }


\begin{document}
\maketitle
%\begin{textblock}{10}(4.3,4.1) \begin{small}Dedicated to Barry Simon on the occasion of his 80th birthday \end{small} \end{textblock}

\begin{abstract}
We deform block triangular Jacobi matrices appearing 
in geometry, look at multi-scale Barycentric limits of geometries and
droplet boundary manifolds in Potts networks.
\end{abstract}

\section{Introduction}

\paragraph{}
A {\bf finite abstract simplicial complex} $G$ of dimension $q$ with exterior derivative $d$
defines a block triangular {\bf Dirac matrix} $D=d+d^*$, a block Jacobi matrix with $q+1$ blocks.
It is the square root of the block diagonal Hodge Laplacian $L=D^2=d d^* + d^* d$. 
For any continuous function $g$, a continuum symmetry appears in the form 
of isospectral {\bf Lax deformations} 
$\frac{d}{dt} D_t=[B_t,D_t]$ with $B_t=g(D)^+-g(D)^-$. This
is a non-linear ODE. We also get $D_t=Q_t^* D Q_t$ from 
the {\bf QR decomposition} $e^{-t g(D_0)}=Q_t R_t$. The ODE and QR pictures agree:

\begin{thm}
ODE and QR give the same isospectral $D_t=c_t+c_t^*+m_t$ and $c_t+c_t^*$ is Dirac.
\end{thm}

\paragraph{}
The {\bf Barycentric refinement} on finite abstract simplicial complexes 
defines a sequence $G=G_0,G_1,G_2,\dots$ of topologically
equivalent complexes. The {\bf density of states} $dk(G_n)$ of the operators
$L(G_n)$ satisfies a {\bf central 
limit theorem} in every dimension $q={\bf dim}(G)$: there is a weak {\bf limiting law} for any $q \geq 1$.
For $q>1$, it appears to be singular continuous. 

\begin{thm}
For $q \geq 1$, the weak limit $\lim_{n \to \infty} dk(G_n)$ exists and is independent of $G$. 
\end{thm}

\paragraph{}
A function $f: V(G) \to K_k=\{0, \dots, k\}$ defines a {\bf network gas} containing
$k$ different particles. If $f(v)=j$ then particle $j$ is at $v \in V$, if $f(v)=0$ no particle 
is at $v$.  One can see $f$ as a configuration in a {\bf Ising} or {Potts model} on the network.
Define the {\bf level set} $G_f=\{x \in G, f(x)=K_k\}$. 
A {\bf $m$-manifold} is a simplicial complex for which unit spheres
are $(m-1)$-spheres.
%where $f((x_0,\dots, x_j))=(f(x_0), \dots, f(x_j))$.

\begin{thm}
If $G$ is a $m$-manifold, then $G_f$ is either empty or a $(m-k)$-manifold. 
\end{thm}

\paragraph{}
Jean Dieudonne \cite{Dieudonne1989} 
singled out {\bf incidence matrices}, {\bf Barycentric refinement}
and {\bf duality} as core innovations in {\bf combinatorial topology}. Theorems 1)-3)
address these pillars: {\bf 1)} incidence matrices are the blocks in the Dirac matrix $D$,
{\bf 2)} Barycentric refinements $G_n \to G_{n+1}$ lead from the discrete to a continuum limit.
{\bf 3)} Duality manifests in a {\bf hyperbolic structure}: every unit sphere $S(x)$ is 
the {\bf join} of dual spheres $S^+(x) * S^-(x)$. 

\paragraph{}
The topic relate to themes in the work of Simon: 
1) is  part of {\bf spectral theory}, {\bf band Jacobi matrices},
scattering e.g. \cite{Cycon}. 2) links to {\bf singular spectra}, 
{\bf multi-scale analysis} and linear algebra \cite{SimonFunctionalIntegration,SimonTrace}.
3) points to {\bf lattice gas} models in statistical mechanics \cite{SimonStatMechanics}:
level sets of functions form boundaries of {\bf droplets}, regions on which spins have 
a fixed value. The topics are illustrated using {\bf computer algebra},
where \cite{SimonNotices1992,SimonPower95,SimonPower97} decades ago compared different CAS's. 

\section{Illustration}

\paragraph{}
The {\bf octahedron} is the $2$-dimensional simplicial complex
\begin{tiny} $G=\{\{1\},\{2\},\{3\},\{4\},\{5\},\{6\}$, \end{tiny}

\begin{tiny} $\{1,3\},\{1,4\},\{1,5\},\{1,6\}$,  
    $\{2,3\},\{2,4\},\{2,5\},\{2,6\},\{3,5\},\{3,6\},\{4,5\},\{4,6\}$, 
    $\{1,3,5\},\{1,3,6\}$, $\{1,4,5\},\{1,4,6\}$, 
    $\{2,3,5\},\{2,3,6\}$, $\{2,4,5\},\{2,4,6\}\}$.
\end{tiny}
Its Hodge Laplacian $L=d d^* + d^* d=L_0 \oplus L_1 \oplus L_2$ contains as blocks
a $6 \times 6$ matrix $L_0$, a $12 \times 12$ matrix
$L_1$ and a $8 \times 8$ matrix $L_2$. 
The Euler characteristic is $\chi(G)={\rm str}(e^{-0 L}) = 6-12+8=2$ with
$b_0-b_1+b_2={\rm dim}({\rm ker}(L_0))-{\rm dim}({\rm ker}(L_1))+{\rm dim}({\rm ker}(L_2)) 
={\rm str}(e^{-\infty L})$. {\bf McKean-Singer} assured $\chi(G) ={\rm str}(e^{-t L})$ for all $t$.

\begin{figure}[!htpb]
\scalebox{0.15}{\includegraphics{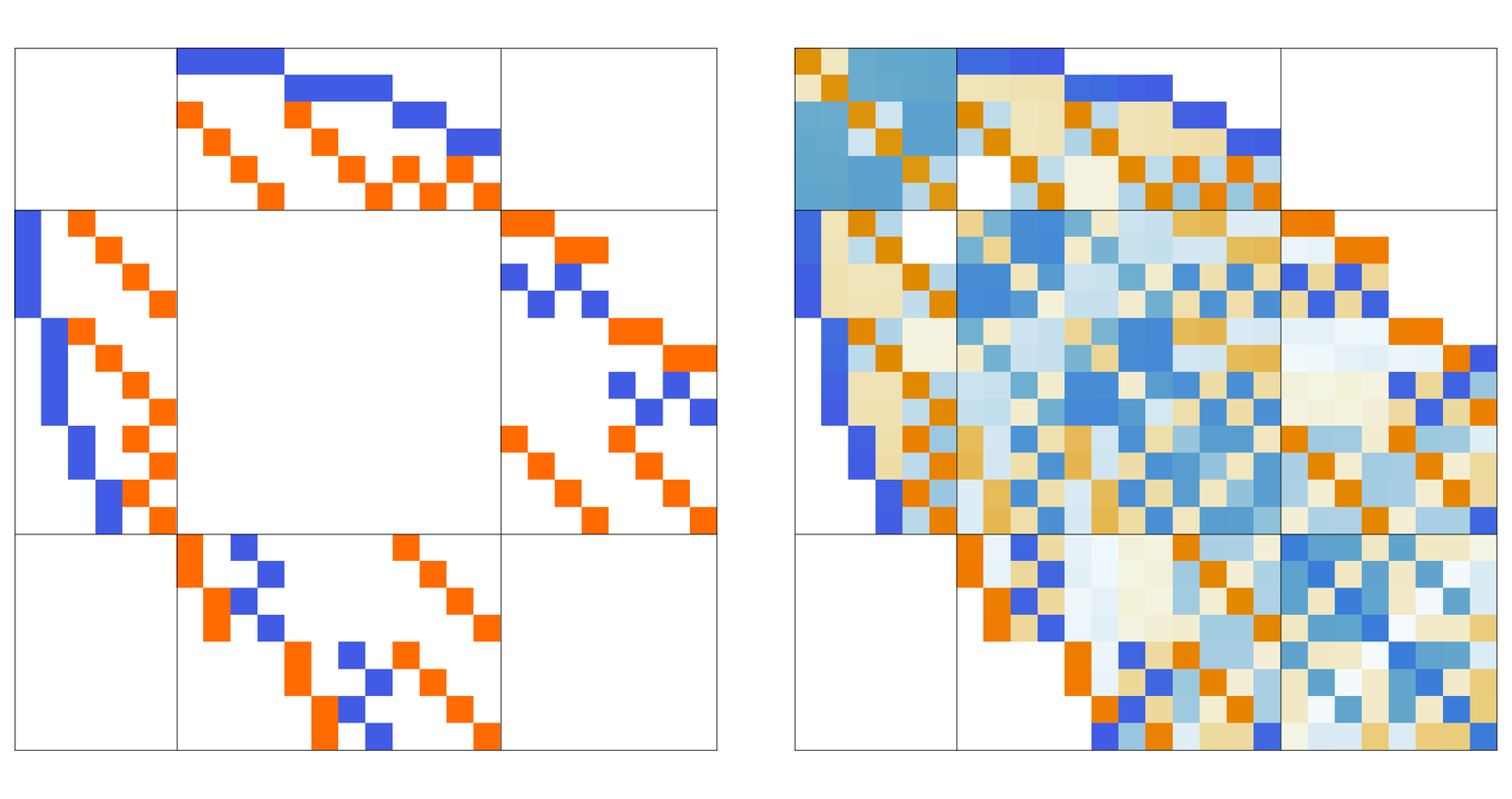}}
\scalebox{0.15}{\includegraphics{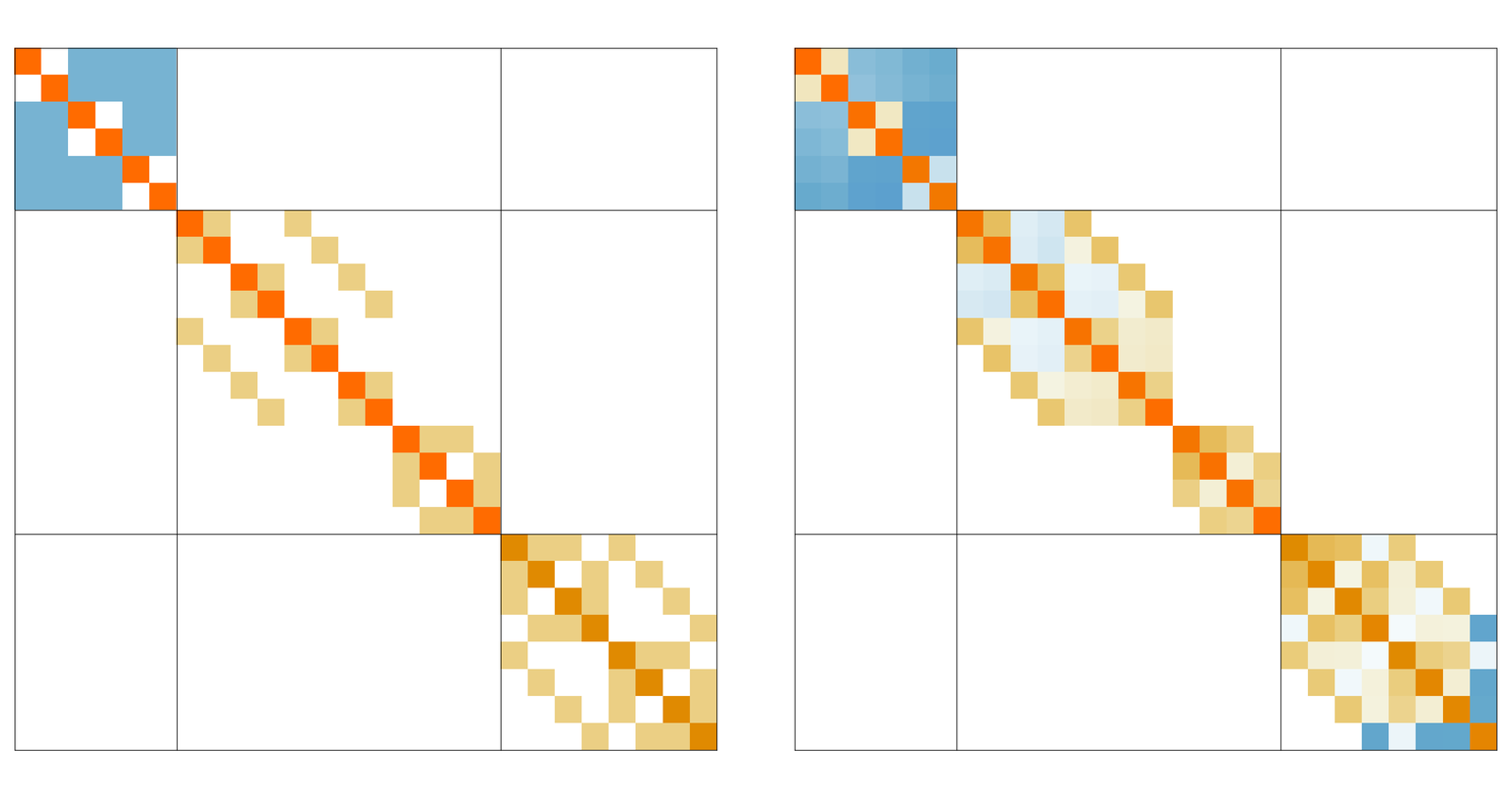}}
\label{Dirac}
\caption{
The matrix $D$ and Laplacian $L$ before and after the deformation. 
}
\end{figure} 

\paragraph{}
We see some initial refinements of the octahedron and then 
the spectrum,integrated density of states and Lyapunov exponent of the Hodge Laplacian of the 4th refinement $G_4$.

\begin{figure}[!htpb]
\scalebox{0.1}{\includegraphics{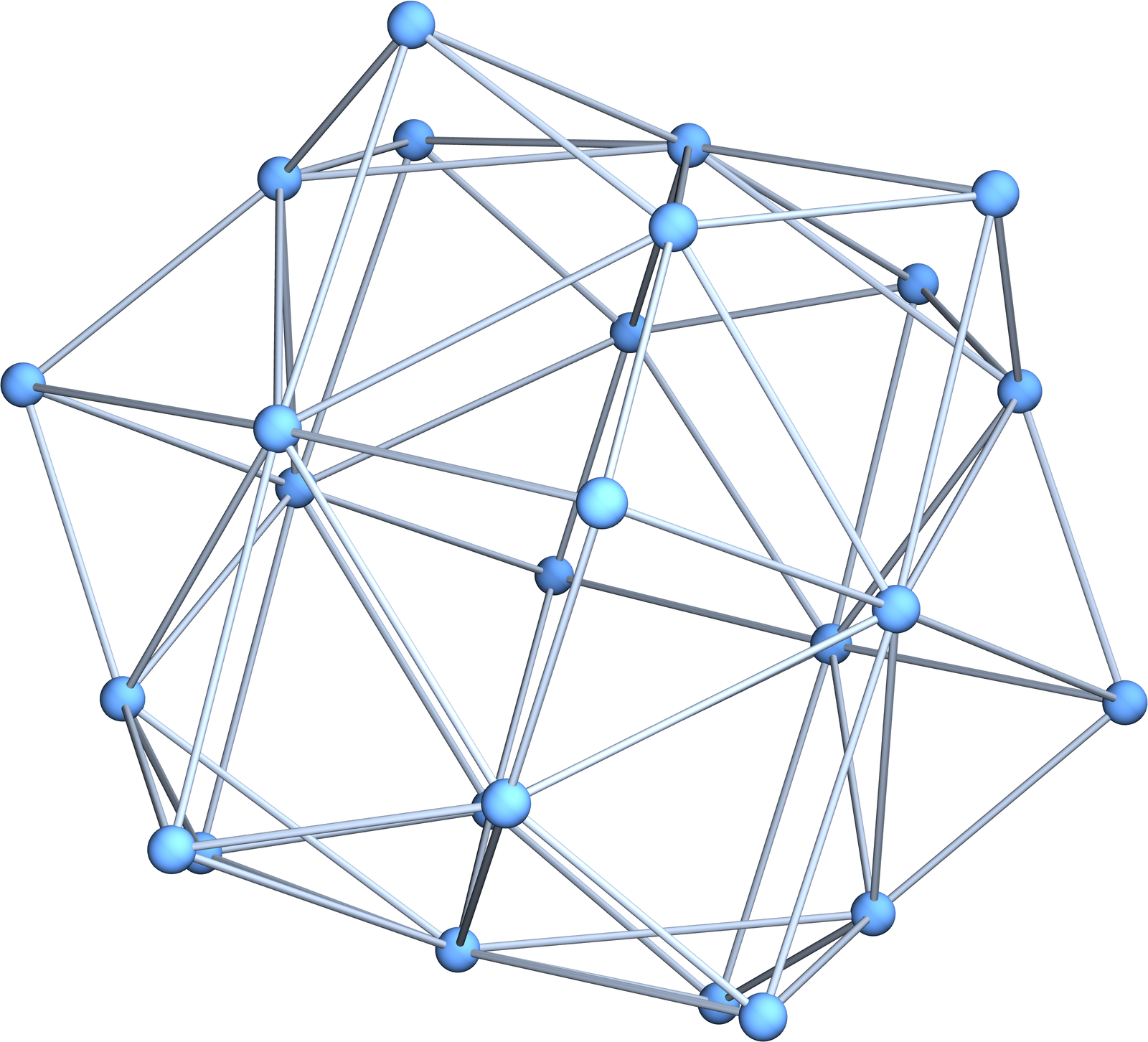}}
\scalebox{0.1}{\includegraphics{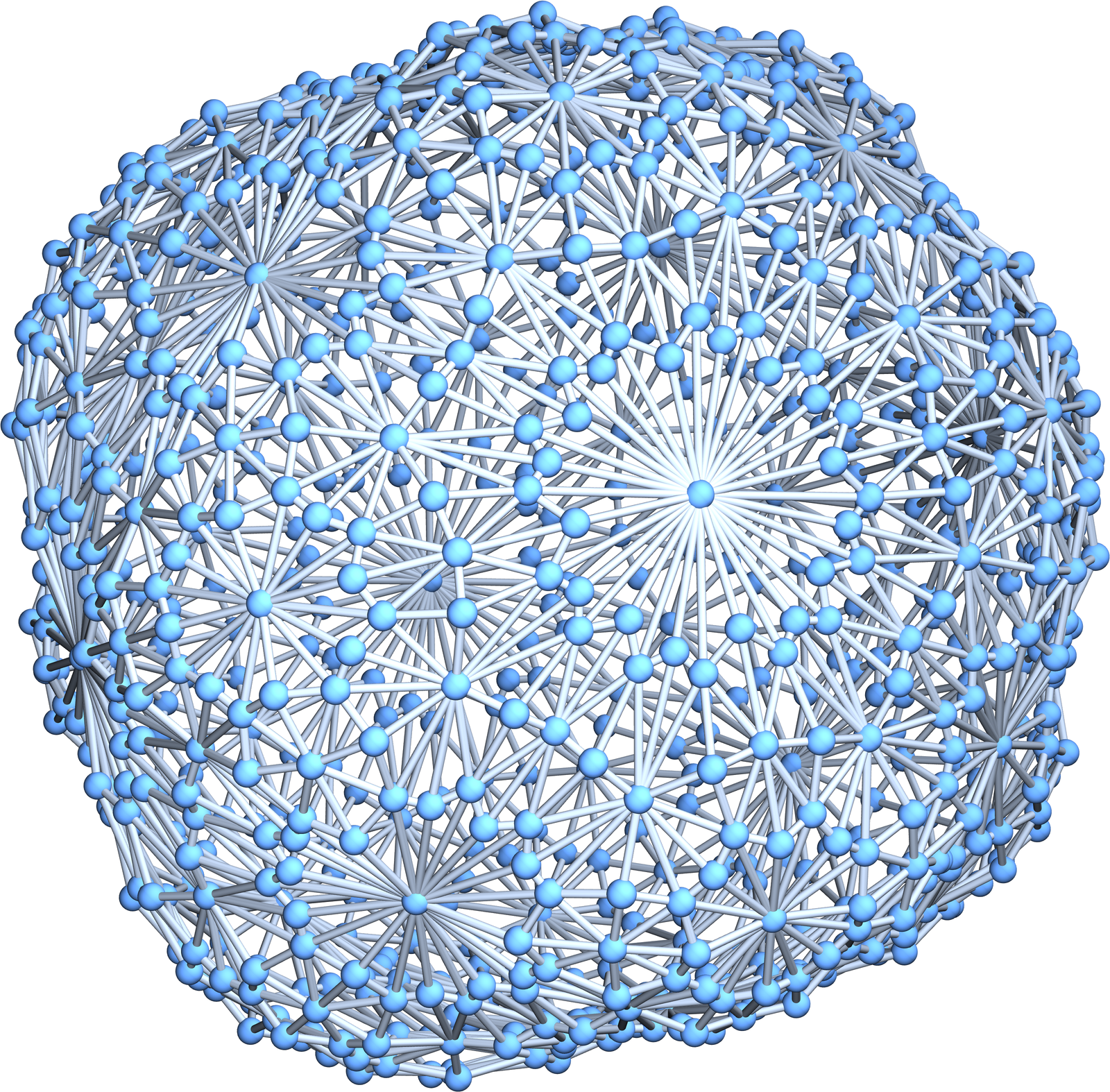}}
\scalebox{0.1}{\includegraphics{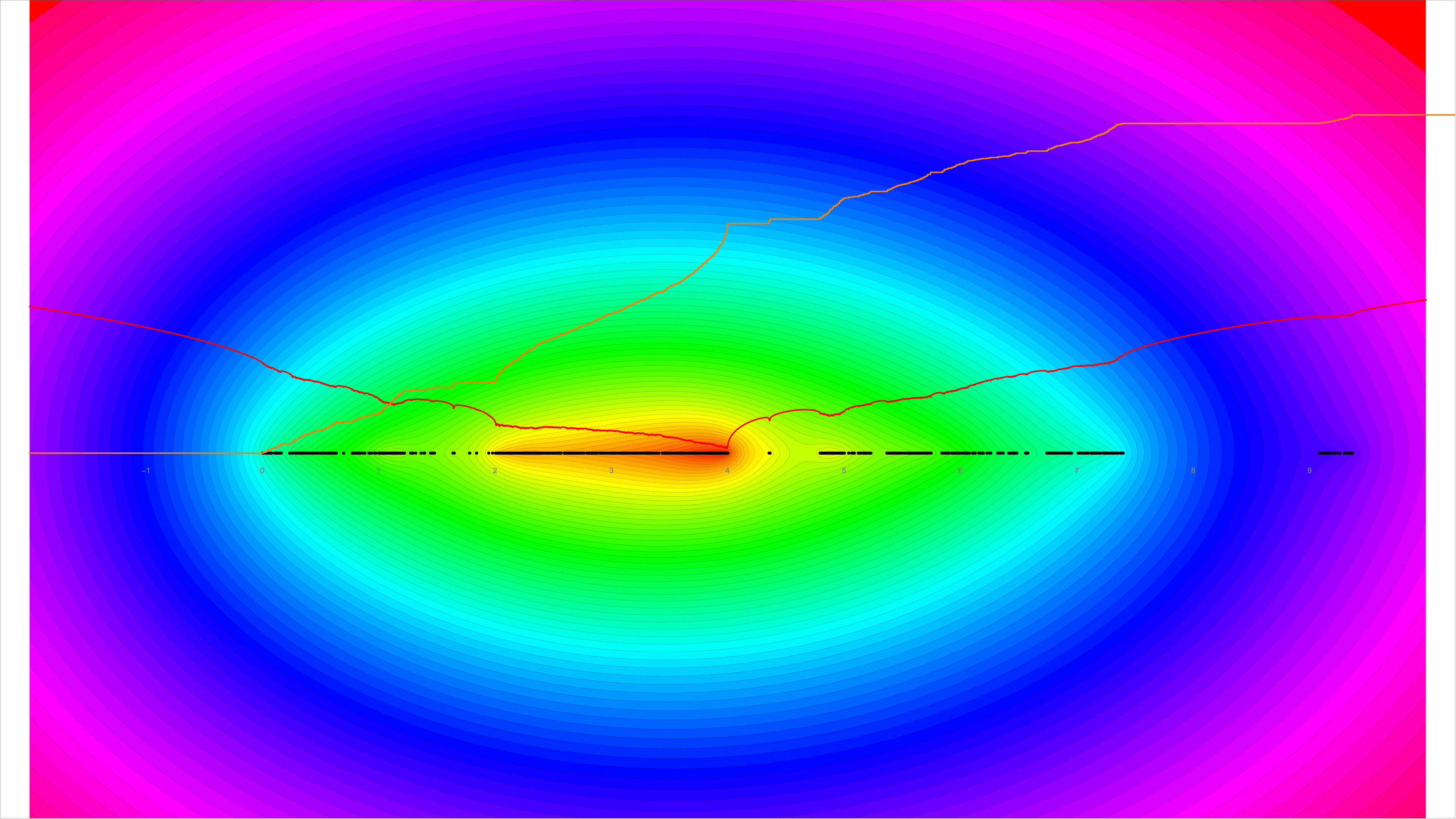}}
\label{Octahedron}
\caption{
Barycentric refinements of $G$ and the Hodge spectrum of $G_4$. 
}
\end{figure}

\paragraph{}
To illustrate the last theme, we pick as a host manifold $G$ a $4$-sphere $S^2 * S^1$, 
where $S^2$ is the {\bf octahedron} and $S^1$ the cyclic graph $C_{15}$ and where $*$ is the {\bf join}.
Now place $k=2$ particle types on some of the vertices of $G$. 
This map $f:V \to \{ 0,1,2 \}$ defines the interface $2$-manifold $G_f$.

\begin{figure}[!htpb]
\scalebox{0.1}{\includegraphics{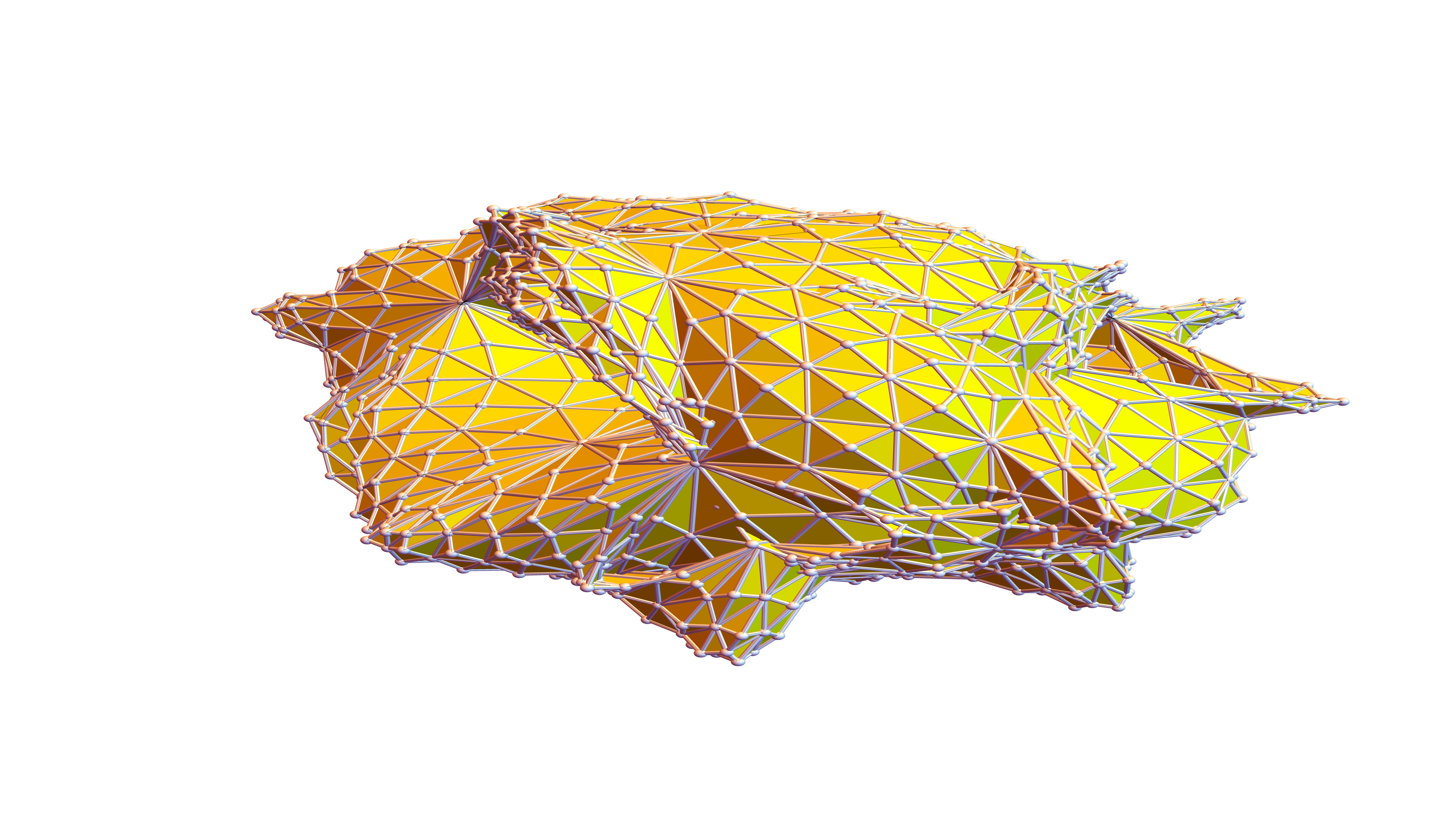}}
\scalebox{0.1}{\includegraphics{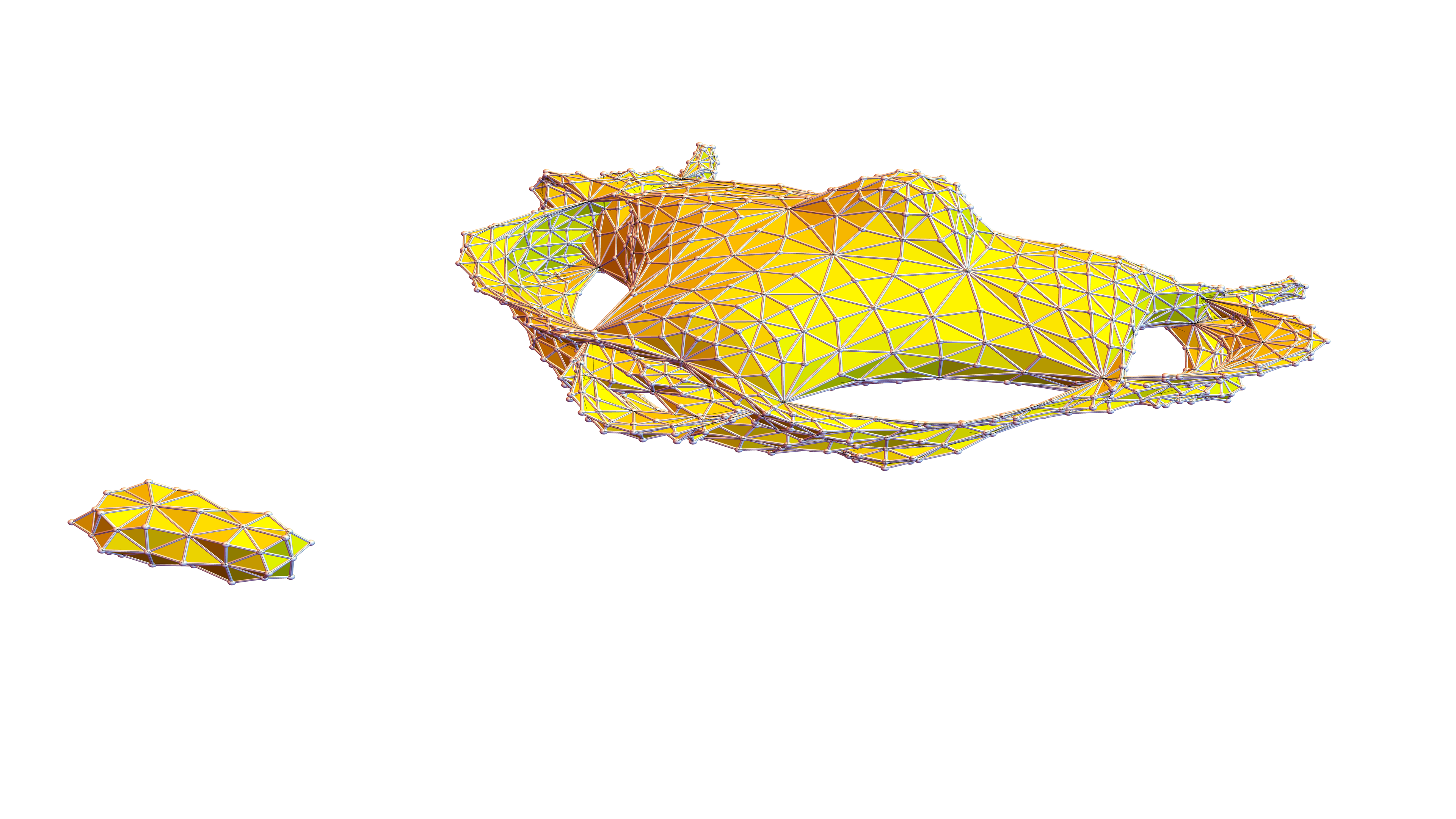}}
\label{Octahedron}
\caption{
Two co-dimension 2 surfaces of a 4-manifold $G$. They are 2-manifolds
and subgraphs of the Barycentric refinement $G_1$. 
}
\end{figure}

\pagebreak

\section{Block Jacobi Matrices}

\paragraph{}
While the first story is {\bf geometric}, it addresses on a central topic in {\bf analysis} \cite{Simon2017}:
given a particular element $L$ in a class of operators defined by some geometry, one can ask
about the structure of the set of operators that are isospectral to $L$.
For finite Jacobi matrices or Sturm-Liouville operators, 
where the underlying geometry is either discrete or continuous intervals or circles,
this relates to {\bf scattering theory} and {\bf orthogonal polynomials}.
For periodic Jacobi matrices, it is linked to algebraic geometry because the 
isospectral set is a torus and points on this torus are encoded by Abel-Jacobi maps. 
Isospectral deformations given by integrable systems play a role there.
\cite{KnillBarycentric2}.

\paragraph{}
If one looks for isospectral sets of Laplacians of higher dimensional 
manifolds, the isospectral set is in general discrete. No deformation is possible in general 
in higher dimensions by a result of Mumford (\cite{Moerbeke79} page 337). 
Sometimes a deformation is possible:
we can for example fix a smooth compact manifold and look at all metrics for which 
the Laplacian $L$ has the same spectrum. For example, there exists a 16-dimensional set
of isospectral flat tori \cite{Milnor64} $\mathbb{R}^{16}/D_{16}^+,\mathbf{R}^{16}/(E_8 \oplus E_8)$
where $D_{16}^+,E_8 \oplus E_8$ are integral unimodular lattices.

\paragraph{}
We gain more flexibility by writing the Laplacian $L$ as a square $L=D^2$
and allow the operator $D=d+d^*$ to develop a diagonal part while keeping a 
block tri-diagonal {\bf block Jacobi structure} $D_t=c_t+c_t^*+m_t$. 
This works also in the continuum but we do not discuss it here, 
as the deformed operators on manifolds would become
{\bf pseudo differential operators} $c_t+c_t*+m$.

\paragraph{}
While $c_t$ is still an exterior derivative, it is not isospectral to the standard one; also non-isospectral
and a prototype of a modification of the exterior derivative is the {\bf Witten deformation} 
\cite{Witten1982,Cycon}. But the new {\bf ``electromagnetic" Dirac matrix} $C_t=c_t+c_t^*$ 
still defines a distance via the {\bf Connes metric} \cite{Connes} which can be used on $G$
as $d_{Connes}(x,y) = {\rm sup}_{f, |Df|_{\infty}=1} |f(x)-f(y)|$.

\paragraph{}
The {\bf exterior derivative} $d$ on a finite abstract 
simplicial complex $G$ defines block triangular Jacobi matrices 
$D=d+d^*$ and a matrix $L=D^2=d d^* + d^* d$. 
Because $d^2=0$ and so $d*^2=0$, it follows that for any polynomial 
$g$, the matrix $g(D)$ has the same Jacobi block structure. The functional calculus
allows to define then $g(D)$ for an arbitrary real continuous function $g$, still
keeping the property that $g(D)$ has a Jacobi block structure. The k'th block size
is the number $f_k(G)$ of simplices of dimension $k$. 
Given such a block Jacobi matrix, one can write it in the form $A=A^+ + A^- + A^0$, where 
$A^+$ are the upper right blocks, where $A^-$ the lower left and where $A_0$ are the diagonal blocks. 
This decomposition reflects that $A^+=(A^-)^*$ maps $l^2(G_k)$ to $l^2(G_{k+1})$ and $A^0$ 
preserves the space $l^2(G_k)$. One obtains so the {\bf anti-symmetric }$B(A)=A^+-A^-$. 
The {\bf Lax differential equations} $D'=[B(g(D)),D]$ can be solved by $D_t=Q_t^* D_0 Q_t$, 
coming from the QR-decomposition $e^{-t g(D)}=Q_t R_t$ of the {\bf heat operator} $e^{-t g(D)}$. 
This was observed first by Symes for the Toda system \cite{Symes} then extended and simplified by 
Deift,Li and Tomei \cite{DeiftLiTomei,Symes}.

\subsection{Finite Spectral Geometry}

\paragraph{}
A {\bf finite abstract simplicial complex} is a finite set $G$ of 
non-empty sets that is closed under the operation of taking 
finite non-empty subsets. If $G$ has $n$ elements, it 
defines a $n \times n$ Dirac matrix $D=d+d^*$, where $d$ is the 
exterior derivative $d: l^2(G_k) \to l^2(G_{k+1})$ with 
$G_k=\{x, |x|=k+1 \}$. The incidence matrices $d_k$ in this {\bf co-chain complex}
are explicitly given as $df(x)=\sum_j (-1)^j f(\delta_j x)$, where $\delta_k:G_{k+1} \to G_k$ 
are the {\bf face maps} of the {\bf delta set} $G=\bigcup_{k=0}^q G_k$. 
The Laplacian $L=D^2=\oplus_{k=1}^q L_q$ is {\bf block diagonal} 
with respect to the decomposition $l^2(G)=\oplus_{k=0}^q l^2(G_k)$. The {\bf k-form Laplacians}
$L_k$ generalize the {\bf Kirchhoff matrix} $L_0= {\rm div}({\rm grad})= d_0^* d_0$ which 
is a discretized Laplacian. 

\paragraph{}
Following Hodge in the continuum \cite{Hodge1933} and Eckmann in the finite case 
\cite{Eckmann1944}, the 
{\bf k'th cohomology} $H^k(G)$ can be identified as the kernel of $L_k$ on $l^2(G_k)$. 
The dimension of the kernel $b_k={\rm ker}(H^k)$ is the $k$'th {\bf Betti number} 
of $G$. This leads to the {\bf Betti vector} $b(G)=(b_0,b_1,\dots, b_q)$.
The number $f_k$ elements in $G_k$ are the components of the {\bf $f$-vector}
$f(G)=(f_0,f_1,\dots f_q)$. The {\bf Euler-Poincar\'e formula} 
$\chi(G)=\sum_k (-1)^k f_k = \sum_k (-1)^k b_k$ can be verified by heat
deformation: ${\rm str}(e^{-tL})$ is for $t=0$ the left hand side
while the limit $t \to \infty$ gives the right hand side. 
The {\bf super trace} of a block diagonal matrix $K=\oplus_{k=0}^q K_k$ is defined as 
$\sum_{k=0}^{q} (-1)^k {\rm tr}(K_k)$. 
As $L=(d+d^*)^2$ is block diagonal, its super trace is defined.
The McKean-Singer symmetry is the statement that the non-zero spectrum of 
$L$ on $\oplus_k l^2(G_{2k})$ is the same than the non-zero 
spectrum on $\oplus_k l^2(G_{2k+1})$.
This implies that ${\rm str}(e^{-tL})$ is time independent and the
scattering argument proves Euler-Poincar\'e. The heat flow morphs the 
combinatorial Euler characteristic to the cohomological Euler characteristic. 
In the deformed situation, the statement still holds and also holds for
$e^{-t C_s}$, where $C_s=(c_s+c_s^*)^2$ is the Laplacian obtained by ignoring the 
diagonal part $m_s$ of $D_s = d_s+d_s^*+m_s$ (we use $s$ for the QR deformation here
because $t$ had been used for the heat deformation). The deformation $D=C_0  \to C_s$
has similar properties than the Witten deformation. The kernels of the blocks of
$C_s$ are still the Betti numbers but the non-zero spectrum of the operators can 
change. 

\begin{figure}[!htpb]
\scalebox{0.2}{\includegraphics{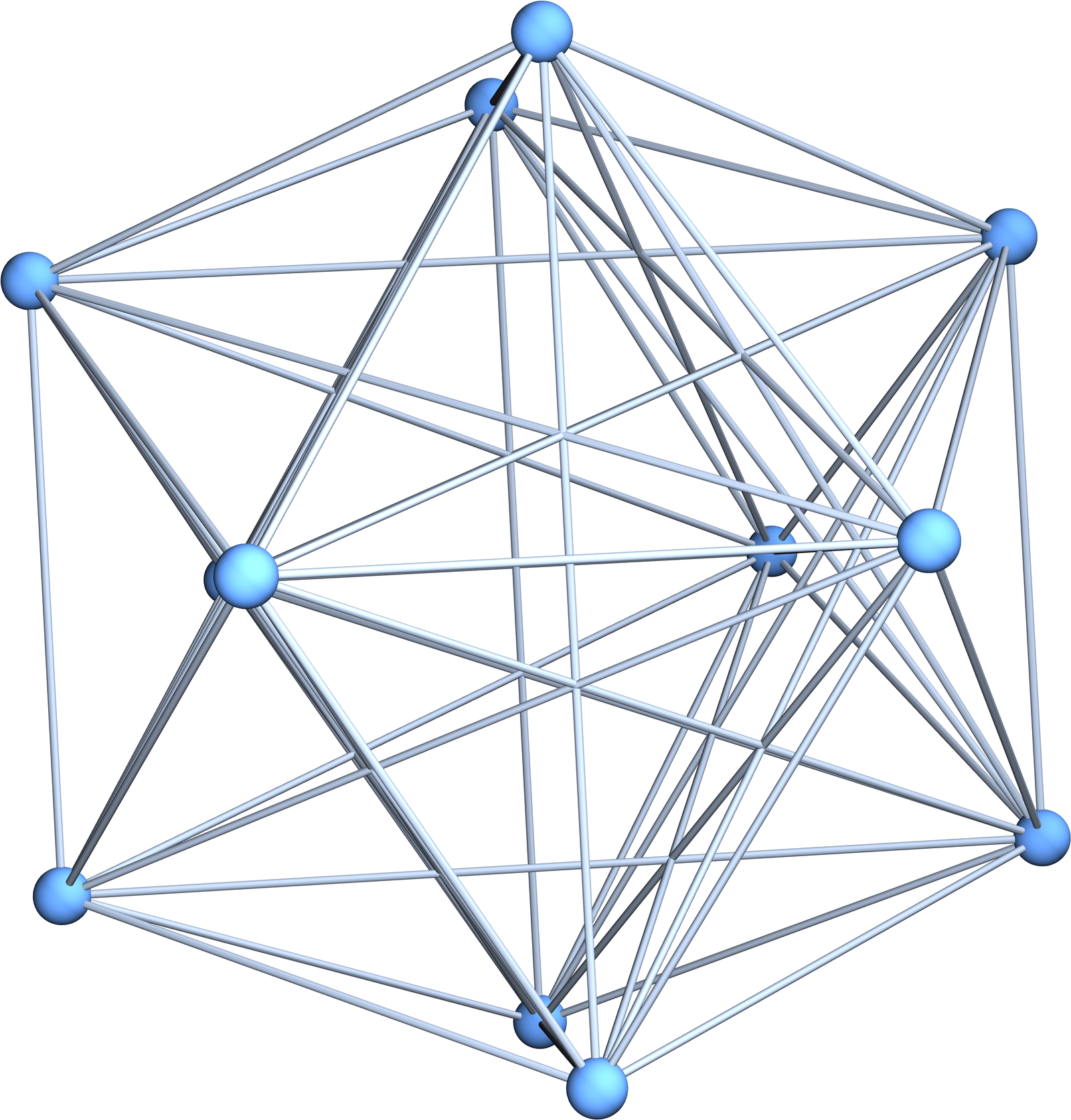}}
\label{5sphere}
\caption{
The {\bf $64$-cell} is the discrete $5$-dimensional cross-polytop. It 
is a discrete 5-sphere. It leads to a simplicial complex $G$ with $f$-vector 
$f(G) = (12, 60, 160, 240, 192, 64)$ and Betti vector 
$b(G) = (1,0,0,0,0,1)$ and Euler characteristic $0$.
}
\end{figure} 

\paragraph{}
The operator $D$ defines a natural metric in $G$, once
an algebra $A$ is fixed, so that one has a {\bf spectral triple} $(G,D,A)$ in the 
sense of \cite{Connes}. A natural choice is $A=l^{\infty}(G)$, the algebra 
of bounded functions on $G$, equipped with the supremum norm. 
The Connes {\bf pseudo distance} 
$d_{Connes}(x,y) = \sup_{f \in A, |[D,f]|_{\infty}=1} |f(x)-f(y)|$
on $G$ becomes a metric after possibly {\bf Kolmogorov identifying}
points that have zero pseudo distance. For a reasonable Dirac matrix $D$,
the condition $d_{Connes}(x,y)=0$ implies $x=y$ so that the Connes distance is already 
a distance. The metric space can so be deformed in a continuous way 
by changing $D$; and this is possible despite the fact that $G$ was finite. The deformation
does not affect other topological properties of $G$ like its topological dimension. 

\subsection{Deformation}

\paragraph{}
Let us look more closely at the mathematics of the 
{\bf isospectral Lax deformation} $\frac{d}{dt} D_t = D'=[B,D]$ with $B=g(D)^+-g(D)^-$.
It satisfies $D_t=Q^*_t D Q_t$ where $d_t=Q^*_t d_0 Q_t$.
We can write have $d_t=c_t+n_t$, where $c_t$ does map $k$-forms to $(k+1)$-forms
and $n_t$ preserves $k$-forms and $m_t=n_t+n_t^*$. We get $D_t= d_t + d_t^* = c_t +c_t^* + m_t$. 

\begin{figure}[!htpb]
\scalebox{0.5}{\includegraphics{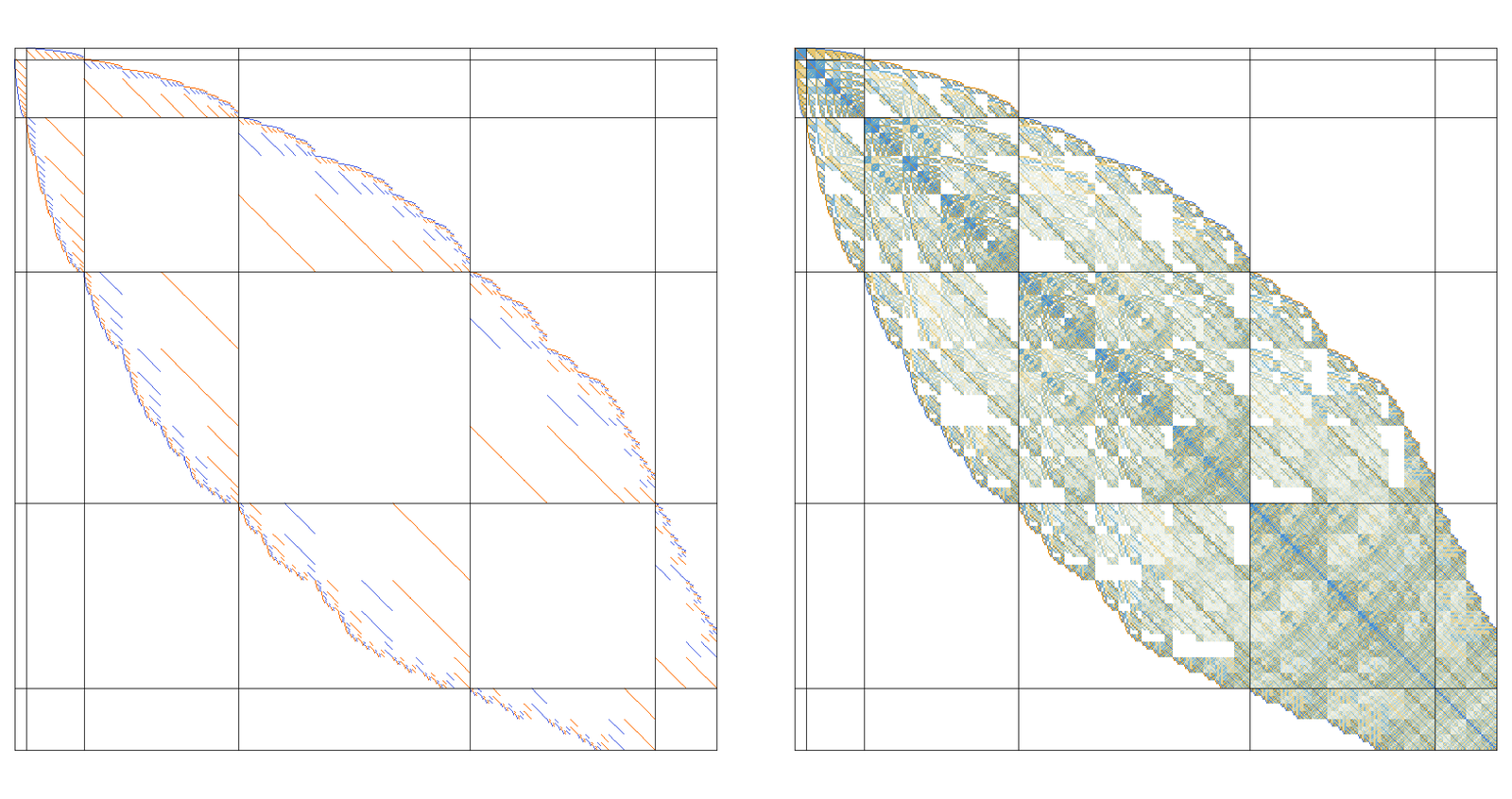}}
\label{Dirac}
\caption{
We see the Dirac matrix $D$ that belongs to the already mentioned {\bf 64 cell}.
It is a $728 \times 728$ block triangular Jacobi matrix
with 6 diagonal blocks. The blocks are the incidence matrices
which Poincar\'e already introduced. The isospectral deformation
$c_t+c_t^*+m_t$ to the right has a block diagonal part $m_t$. The
blocks $c_t$ define new exterior derivatives, leading to
to operators $C_t^2=(c_t+c_t^*)^2$ with the 
same Betti numbers. Like in the Witten deformation, the spectrum 
of $C_t$ changes but the kernel of $C_t$ does not. 
}
\end{figure} 

\paragraph{}
The {\bf Toda chain} (see \cite{Toda1967,Toda}) is a {\bf Lax pair} \cite{Lax1968}
$L'=[B,L]$ with $B=a-a^*$ and block tri-diagonal $L=a+a^*+b$. 
If $L$ is invertible, it is possible to write $L \oplus L_1=D^2$, 
where $D$ is defined on a doubled lattice and $L_1$ is a 
B\"acklund transformation of $L$. We are here in a much more general situation
as we can deform Laplacians of geometries in any dimension. 
The deformation applies for general geometries and also for compact Riemannian manifolds 
\cite{IsospectralDirac2}. 

\paragraph{}
As shown in the Lax case \cite{DeiftLiTomei,Symes}, the flow can 
be integrated using the QR-flow. The proof of Deift,Li and Tomei 
\cite{DeiftLiTomei} works also here. Let us reformulate Theorem~1 as: \\

\begin{center} \fbox{\parbox{10cm}{
If $\exp(-t g(D_0))=Q_t R_t$, 
then $D_t=Q^*_t D Q_t$ solves $D' = [B,D]$, with $B=g(D)^+-g(D)^-$. 
}} \end{center} 

\begin{proof} 
Look first at the QR picture:
Let $Q_t,R_t$ be obtained from the QR decomposition of the block triangular 
$e^{t g(D_0)}= Q_t R_t$. The Gram-Schmidt procedure shows that $Q_t$ is 
block triangular. Since $e^{t g(D_0)}$ is invertible, the QR decomposition is
unique. Define $D_t=Q_t^* D_0 Q_t$. We have now $(Q_t,R_t,D_t)$.  
Now to the differential equation picture: look at
$$   R'_t=[2g(D_t)^+ + g(D_t)^0] R_t,   Q'_t=Q_t [g(D_t)^--g(D_t)^+] $$
for the pair $R_t,Q_t$ with $R(0)=Q(0)=1$ and let $D_t=Q_t^* D_0 Q_t$ so
that $D_t'=[g(D_t)^--g(D_t)^+,D_t]$. 
The equations imply that $R_t$ remains upper triangular and 
that $Q_t$ remains orthogonal and $D_t$ block triangular. We have 
so an other way to describe $(Q_t,R_t,D_t)$. 

To see that these two pictures are compatible we show that
$H(Q,R) = e^{t g(D_0)} - Q_t R_t$ is an integral of motion of the 
differential equation: just differentiate
$\frac{d}{dt} e^{t g(D_0)} = g(D_0) Q_t R_t$ and
$$ \frac{d}{dt} (Q_t R_t) = 
   Q_t [g(D_t)^--g(D_t)^+] R_t + Q_t [2g(D_t)^+ + g(D_t)^0 ] R_t
                       = Q_t g(D_t) R_t  \; . $$
Since $R_t$ is invertible, 
$g(D_0) Q_t R_t = Q_t g(D_t) R_t$ follows from $Q_t^* g(D_0) Q_t = g(D_t)$.
We have shown that $\frac{d}{dt} H(Q(t),R(t)) = 0$. The differential equation
is compatible with the QR computation. 
\end{proof} 

\paragraph{}
Changing $g$ produces different flows. 
All these flows commute like in the Toda case \cite{Symes}. 
The dimensions of the kernels of $(c+c^*)^2$ agree
with the dimensions of the kernels of $(d+d^*)^2$. 
One can see this by noting that if $Du=(d+d^*)u=0$, and $u_t=Q u$,
then $(D u)' = D' u + D u' = [D,B] u + D Q' u = D B u - B D u - D B u = -B Du$
which is zero whenever $Du=0$. Differentiate again to see $(Du)''=-B'Du + B (Du)'=0+0$
which is again zero etc.

\paragraph{}
The {\bf McKean-Singer formula} still holds also in the deformed case. The 
super trace ${\rm str}(L^n)=0$ of any power $L_t^k$ of $L_t=D_t^2$. This implies
the McKean-Singer formula ${\rm str}(e^{-L_t}) = \chi(G)$. 
Comparing $t=0$ and $t=\infty$ establishes the {\bf Euler-Poincare equations}
$\chi(G) = \sum_k (-1)^k b_k(G)$. 

\paragraph{}
The deformation also works in division algebras like 
$\mathbb{C}$ or $\mathbb{H}$. To do the deformations there, look at 
$$ D' = [ g(D)^+- \overline{g(D)}^-+ i \beta g(D)^0, D]   $$
where $i$ is an imaginary unit or a unit quaternion. 
One can still associate with this a 
$QR$ decomposition by modifying the differential equations for $R,Q$
$$   R'_t=    [2g(D_t)^+ + g(D_t)^0 - i \beta g(D)^0] R_t ,   
     Q'_t=Q_t [\overline{g(D_t)}^--g(D_t)^+ + i \beta g(D)^0   ] $$
Now, $Q(t)$ is unitary and $R(t)$ block upper triangular complex. 
In the limit $t \to \infty$, it aligns more and more to the Schr\"odinger 
wave equation $Q_t' = i m_t \beta Q_t$. 

%  PROOFREAD SO FAR. CLARIFY THE MAIN PROOF. 

\section{Singular spectra} 

\paragraph{}
The second story deals with a {\bf scaling symmetry} in geometry \cite{KnillBarycentric2}.
The refinement step is the {\bf Barycentric refinement}. Scaling considerations
are ubiquitous in mathematical physics. Almost periodicity for example 
manifest, if smaller parts get expanded to larger structures through
substitutions leading to geometries with {\bf self-similarity}. This also works in discrete
settings like {\bf substitution dynamical systems} \cite{Queffelec}.
Multi-scale analysis appears in solid state physics, for example, when
looking at the mathematics of {\bf Anderson localization}, the phenomenon of 
random systems to display point spectrum. See for example \cite{FroehlichSpencer}.

\paragraph{}
Scaling symmetries appear in other parts of mathematical physics. 
They are familiar on {\bf unimodular maps} $f \in C^2([-1,1],[-1,1])$ 
satisfying $f(0)=1$ $f''(x)<0$ and $\lambda=f(1)<0$. 
The {\bf Feigenbaum-Cvitanovi\'c functional equation} has a fixed point of
$T(g)(x) = \frac{1}{\lambda} g(g(\lambda x))$, explaining universality. 
See for example \cite{Lanford84}. 

\paragraph{}
Similarity symmetry also appears in probability theory. There is the
renormalization map $X \to T(X)=\overline{X+X}$ on the Hilbert space of $L^2$ random variables, 
where $\overline{Y}$ is the {\bf centered normalized version} $(Y-{\rm E}[Y])/\sigma(Y)$ of 
{\bf mean} $0$ and variance $1$. The {\bf Gaussian distribution} is a fixed point of $T$. 
The fact that $T^n(X)$ converges to this fixed point for any initial $L^2$ random variable 
is known as the {\bf central limit theorem.}

\subsection{Barycentric refinements}

\paragraph{}
The {\bf Barycentric refinement} of a finite abstract simplicial complex $G$
is a new complex $G_1 = \{ x_1 \subset G, \forall x,y \in x_1$ either $x \subset y$ or $y \subset x \}$.
If $G$ is the Whitney complex of a graph, the Barycentric refinement comes from
the refined graph, where the complete subgraphs are the vertices and where two new vertices are connected if 
one is contained in the other. The Barycentric refinement of a graph 
$C_n$ for example is $C_{2n}$. The Barycentric refinement of the triangle $K_3$ 
is a wheel graph. 

\paragraph{}
An other example of deformation is {\bf soft Barycentric refinement} \cite{SoftBarycentric}. 
It has the property that in dimension $q=2$, the maximal vertex degrees stay bounded. 
The soft refinement avoids simplices of dimension 
$(q-1)$ that are contained in two or more $q$-dimensional simplices. When we start with 
a triangle $G=K_3$ for example, all the $G_n$ are part of a {\bf hexagonal lattice}. The interior
points have all vertex degree 6 and so are flat points. When we start with a tetrahedron 
$G=K_4$, then we have graphs $G_n$ which are $3$-balls. Their boundary surface is a 2-sphere
which undergoes a Barycentric refinement \cite{SoftBarycentric} which parallels the situation
reviewed here. In two dimensions, we get a piecewise smooth limiting density of states with
a {\bf van Hove singularity}.
% Import["~/text3/graphgeometry/art-simon/code.m"];
% s = CompleteGraph[{2, 2, 2, 2}]; s1 = SB[s, 2]; Union[Regge3[s1]]

\subsection{Density of states}

\paragraph{}
Given an operator on $G$ with a discrete set of non-negative eigenvalues $\lambda_i$,
the {\bf spectral function} $F_G(x)=\lambda_{[n x]}$ encodes the eigenvalues
$\lambda_1 \leq \dots \leq \lambda_n$. To make this work for $x \in [0,1]$, we always assume 
$\lambda_0=0$ (and all other eigenvalues are non-negative). 
The function $F$ is piecewise constant, jumps at eigenvalues, is
monotone and $F_G(0)=0$ and $F_G(1)=\lambda_n$.
The {\bf integrated density of states} $F^{-1}$ is a monotone $[0,1]$-valued 
function on $[0,\infty)$ with the understanding that it is piecewise
constant so that its {\bf distributional derivative}
is the {\bf density of states} probability measure $dk=(F^{-1})'$ on 
$\mathbb{R}^+$, a discrete pure point measure with support on the
spectrum of $L$. Point-wise convergence of $F$ implies point-wise convergence of 
$F^{-1}$ and so weak-* convergence of the density of states measure. 
The function $F^{-1}$ would be called {\bf cumulative distribution function} CDF
in probability theory. 

\paragraph{}
Since we will look at sequences of Barycentric refinements, where the first one is already
the Whitney complex of a graph, we can assume without loss of generality that the 
simplicial complex $G$ comes from a graph $\Gamma=(V,E)$:
the individual simplices are then the vertex sets of complete subgraphs of $\Gamma$. 

\subsection{Distance} 

\paragraph{}
For two graphs on the same vertex set, define the {\bf graph distance}
$d(G,H)$ as the minimal number of edges which need to be modified in 
order to get from $G$ to $H$. If no a priori host graph is given,
let $n(G,H)$ as the minimum number of vertices which a host graph containing both $G,H$
can have. The distance function $d(G,H)$ can be extended to graphs which do not 
have the same vertex set; we just define $d(G,H)$ as the minimal possible distance
we can get by embedding both in a common graph with $n(G,H)$ elements.

\paragraph{}
Having the graph distance defined as the minimal cardinality of the symmetric 
difference $E(G) \Delta E(H)$ if $G,H$ are both subgraphs of $K_n$ one has:

\begin{lemma}
$||F_G-F_H||_1 \leq 4 d(G,H)/n(G,H)$.
\end{lemma}
\begin{proof} 
Kirchhoff Laplacians satisfy $\sum_{i,j} |L_{ij}-K_{ij}| \leq 4 d(G,H)$
because each edge $(i,j)$ affects only the four matrix
entries $L_{ij}, L_{ji},L_{ii},L_{jj}$. 
For two subgraphs $G,H$ of a common graph with $n$ vertices, 
and Laplacians $L,H$, and eigenvalues $\alpha=(\alpha_1, \dots, \alpha_n)$
and $\beta=(\beta_1, \dots, \beta_n)$, the inequality gives
$||\alpha-\beta||_1 \leq 4 d(G,H)$ so that $||F_{G}-F_{H}|| \leq 4 d(G,H)/n$
\end{proof} 

\paragraph{}
For the Hodge Laplacians, the situation is a bit more technical. On each $k$-form sector,
the constant 4 needs to be replaced by a constant that depends on the maximal 
k-simplex degree of $G$. But this constant $C_k(G,H)$ grows exponentially slower than 
$d(G,H)$. We still have $||F_{G}-F_{H}|| \leq C_k(G,H)  d(G,H)$ if we define $d(G,H)$ 
as the graph distance between graphs $\Gamma(G)=(G,E(G))$ and $\Gamma(H)=(H,E(H))$, 
where $E(G) = \{ (x,y) \in G^2, x \subset y$ or $y \subset x$ and $x \neq y \}$. We
focus for simplicity on the 0-form Laplacian. 

\subsection{Linear algebra}

\paragraph{}
The {\bf Lidskii-Last inequality} for two $n \times n$ matrices $A,B$ with eigenvalues
$\alpha_1 \leq  \cdots \leq \alpha_n$ and $\beta_1 \leq \cdots \leq \beta_n$ is

\begin{lemma} $||\alpha-\beta||_1 \leq \sum_{i,j=1}^{n} |A-B|_{ij}$
\cite{SimonTrace,Last1995} for any two complex self-adjoint $n \times n$ matrices $A,B$ with
eigenvalues $\alpha_1 \leq \alpha_2 \leq \dots \leq \alpha_n$ and
$\beta_1 \leq \beta_2 \leq \dots \leq \beta_n$.  
\end{lemma}

\begin{proof}
The formulation as well as the
reduction to the standard Lidksii theorem is due to Last \cite{Last1995}:
denote with $\gamma_1 \leq \cdots \leq \gamma_n$ denote the eigenvalues of the 
selfadjoint matrix $C:=A-B$ and let $U$ be the unitary matrix diagonalizing $C$
so that ${\rm Diag}(\gamma_1, \dots ,\gamma_n)=UCU^*$:
\begin{eqnarray*}
 \sum_i |\gamma_i| &=& \sum_i (-1)^{m_i} \gamma_i
                    = \sum_{i,k,l} (-1)^{m_i} U_{ik} C_{kl} U_{il} \\
                    &\leq& \sum_{k,l} |C_{kl}| \cdot
                                    |\sum_i (-1)^{m_i} U_{ik} U_{il}|
                    \leq \sum_{k,l} |C_{kl}| \; .
\end{eqnarray*}
Lidskii's inequality \cite{SimonTrace} gives 
$\sum_j |\alpha_j-\beta_j| \leq \sum_j |\gamma_j|$. Combining the two
things gives $\sum_j |\alpha_j-\beta_j| \leq \sum_{k,l} |C_{kl}|$. 
\end{proof}

\paragraph{}
As an application, we see that the $l^1$ distance between the Kirchhoff spectra
of two graphs $G,H$ can be measured by their distance:

\begin{coro}
$\sum_j |\lambda_j(G)-\lambda_j(H)| \leq C d(G,H)$.
\end{coro}
\begin{proof}
If $A,B$ are Kirchhoff matrices of two graphs $G,H$ of distance $d(G,H)$, 
then $\sum_{k,l} |A_{kl}-B_{kl}| \leq 4 d(G,H)$. Now use the above lemma. 
\end{proof}

\subsection{Central limit}

\paragraph{}
The statement of Theorem~2 is restated as: \\

\begin{center} \fbox{\parbox{10cm}{
For any $G$, weak limit $\lim_{n \to \infty} dk(G_n)$ exists and is unique and 
only depends on the maximal dimension $q$ of $G$. 
}} \end{center} 

\begin{proof}  % From Univerality  KnillBarycentric2
{\bf (i)} If $G$ has maximal dimension $q$, 
there exists an explicit $(q+1) \times (q+1)$ matrix $A$ 
which maps the $f$-vector of $G_m$ to the $f$-vector of $G_{m+1}$. 
It is {\bf upper triangular} with $A_{ij} = S(i,j) j!$, where $S(i,j)$ are the
{\bf Stirling numbers} $S(i,j)$ of the second kind. The diagonal entries are $A_{kk}=k!$.
It follows that there exists a constant $C$ such that 
$f_q(G_0) ((q+1)!)^m/C \leq f_q(G_m) \geq f_q(G_0) C ((q+1)!)^m$ and
$f_k(G_0) (k+1)!)^m/C \leq f_k(G_m) \geq f_k(G_0) C (k+1)!)^m$ for every $0 \leq k \leq q$. \\
{\bf (ii)} If $G=K_{q+1}$ is the complete graph, then $F_{G_m}$ converges in $L^1[0,1]$ and 
point-wise in $(0,1)$ Proof: For every $m$, there are $(q+1)!$ sub-graphs of $G_{m+1}$ 
isomorphic to $G_m$, forming a cover of $G_{m+1}$. 
Any of their intersections in a lower dimensional set. Since the number of
vertices of this intersection grows with an upper bound $C q!^m$,
we have $d(G_{m+1},\bigcup_j G_m^{(j)} ) \leq (q+1)! C (q!)^m$. This shows
$||F_{G_m}-F_{G_{m+1}}||_1 \leq C (q+1)! (q!)^m/((q+1)!)^m = C/(q+1)^{m-1}$.
This Cauchy sequence in $L^1([0,1])$ has a limit in that Banach space. \\
{\bf (iii)} The Barycentric refinement is a contraction in a suitable metric.
Every $G_m$ is part of a finite union of graphs $H_m$ subgraphs isomorphic to $G_{m-1}$ and
a rest graph $L$ of dimension smaller than $q$. 
Then $F_{G_m}$ and $F_{H_m}$ have the same limit because $F_{H_m}=F_{G_{m-1}}$.
The Barycentric evolution of the boundary of refinements of $K_{q+1}$ as well as $L$ grows
exponentially smaller. Given two complexes $A,B$, then
$d(A_m,B_m) \leq C_q (d!)^m$ so that $||F_{A_m}-F_{B_m}||_1 \leq 4 C_q/(d+1)^m$. 
Let $m_0$ be that $c=C'/(q+1)^{m_0}<1$. 
Define a new distance $d'(A,B) = \sum_{k=0}^{m_0-1} d(A_k,B_k)$,
so that $d'(A_1,B_1) \leq c d'(A,B)$. Apply the {\bf Banach fixed point theorem}. \\
{\bf (iv)} The functions $F_{G_m}$ converge uniformly on every compact sub-interval 
of $(0,1)$: since each $F_{G_m}$ is a monotone function in $L^1([0,1])$,
the exponential convergence on compact subsets of $(0,1)$ follows from
exponential $L^1$ convergence. (See \cite{LewisShisha}).
In dimension $q$, one has $||F_{G_m}||_1 \to (q+1)!$ exponentially fast.
Indeed, as the number of boundary simplices grows like $(q!)^m$ and the number
of interior simplices grows like $((q+1)!)^m$, the convergence is of the order 
$1/(q+1)^m$.  By the {\bf Courant-Fischer mini-max principle}
$F_{G}(1) = \lambda_{n} = {\rm max} (v,Lv)/(v,v) 
  \geq {\rm max}(L_{xx}) = {\rm max}({\rm deg}(x))$
grows indefinitely, so that the $L^{\infty}$ convergence on compact 
sub-intervals of $(0,1)$ can not be extended to $L^{\infty}[0,1]$. 
\end{proof} 

\paragraph{}
It follows that the density of states of $G$ converges ``in law" to a 
universal density of states which only depends on $q$, hence the name ``central limit theorem".
The analogy is to think of $G$ or its Laplacian $L$ as the random variable and the spectrum
$\sigma(L)$ of the Laplacian $L$ as the analog of the probability density and of the
Barycentric refinement operation as the analogue of adding and normalizing two independent
identically distributed random variables.

\section{Manifolds}

\paragraph{}
The third and last story was developed over the last 10 years
\cite{KnillSard,DiscreteAlgebraicSets,DiscreteAlgebraicSets2}.
It is a mixture of discrete geometric analysis, discrete calculus or discrete 
algebraic geometry. It can be seen as a discrete version of the {\bf theorem of Morse-Sard} 
\cite{Morse1939,Sard1942}.
It also touches upon statistical mechanics, if one reinterprets a function taking finitely 
many values on a network as a spin configuration model.
In statistical mechanics, one is interested in thermodynamic quantities given 
by expectations and usually takes an infinite volume limit. We also just look at interface 
structures where all possible function values are attained simultaneously on a simplex. This
produces submanifolds if the host structure is a manifold. 

\paragraph{}
We formulate things in the class of {\bf discrete manifolds}. 
This can be done without any reference to the continuum. Manifolds
are invariant under Barycentric refinements. What is the  analog
of a {\bf variety}? In a classical algebraic geometry frame work, we look for
solutions of systems of polynomial equations in Euclidean or projective host spaces. In the discrete, 
we can just look at maps into a finite set of $(k+1)$ elements and then ask for all simplices on which 
all color values are taken. For $k=1$ for example, this means that we look at all simplices,
where both values $0$ and $1$ are attained. These simplices form {\bf interface cases}.

\paragraph{}
The surprise is that whatever function $f:V(G) \to \{0, \dots, k\}$ we chose on
the $q$-manifold $G$, the level surface is either a $(q-k)$-dimensional manifold if it
is not empty. A priori, these level surfaces $H$ are open sets in $G$ (meaning that if 
$x \in H$ and $x \subset y$ then $y \in H$). But they become again closed sets in 
the Barycentric refinement (meaning sub-simplicial complexes.

\subsection{Topology}

\paragraph{}
On any finite abstract simplicial complex $G$, there is a finite, although 
non-Hausdorff, topology on $G$ which honors connectivity and dimension properties of
$G$. It is defined by the {\bf poset structure} on $G$ given by inclusion $x \subset y$.
This {\bf Alexandrov topology} \cite{Alexandroff1937}
$\mathcal{O}$ on $G$ has as a {\bf topological basis}
the set of {\bf stars} $U(x) = \{ y \in G, x \subset y\}$
which are the smallest open sets containing $x$. In positive dimension $q>0$, 
this topology is never Hausdorff because for an edge $x=(a,b)$, the 
smallest open sets $U(\{a\}), U(\{b\})$ containing $\{a\},\{b\}$ 
always have the intersection $U(x)$. For more, see \cite{FiniteTopology}.

\paragraph{}
The Alexandrov topology on $G$ comes with a hyperbolic structure and duality.
The {\rm cores} $C(x) = \{ y \subset x \}$ are subsimplicial complexes and so closed.
The {\bf stars} $U(x) = \{ y \in G, x \subset y\}$ are open. The complement
of an open set is closed and vice versa. The two sets $U(x)$ and $C(x)$ split the 
unit sphere $S(x)=\overline{U(x)} \setminus U(x)$ into 
a {\bf stable sphere} $S^-(x)=C(x) \setminus U(x)$ and an {\bf unstable sphere}
$S^+(x)=U(x) \setminus C(x)$. Indeed, $S(x)=S^-(x) \star S^+(x)$,
where $A * B$ is the {\bf join} $A \cup B \cup \{ x \cup y, x \in A, y \in B \}$.
A function $f: G \to K_k=\{0,\dots,k\}$ defines an open {\bf level set}
$G_f=\{ x \in G, f(x) = K_k \}$. The name hyperbolicity comes from the fact that
a hyperbolic fixed point $x$ of a smooth map on a manifold has stable and unstable
manifolds such that any small enough sphere $S_r(x)$ is homeomorphic to 
$S_r^+(x) * S_r^-(x)$, where $S_r^{\pm}$ are the intersections of the stable or
unstable manifolds with $S_r(x)$.

\subsection{Discrete Manifolds}

\paragraph{}
We can inductively define a {\bf discrete q-manifold} as a finite abstract simplicial complex
for which all unit spheres $S(x)$ are $(q-1)$-spheres. A $q$-sphere is a $q$-manifold
that is contractible, if one vertex is taken away. A complex $G$ inductively is 
called {\bf contractible} 
if there exists a simplex $x$ such that both $G \setminus U(x)$ and $S(x)$ are contractible. 
The void $\emptyset$ is assumed to be a $(-1)$-sphere and the 
one-point complex $K=\{ \{1\} \}$ is assumed to be contractible. 

\paragraph{}
All this can be reformulated in graph theory and essentially equivalent
because a graph defines a Whitney complex and a complex defines a graph. 
The Whitney complex of the graph of a complex is the Barycentric refinement. 
A {\bf discrete $q$-manifold} is a finite simple graph $G=(V,E)$ 
for which every unit sphere $S(v)$ is a discrete $(q-1)$-sphere. 
For a vertex $v \in V$, the {\bf unit sphere} $S(v)$ is the subgraph
generated by the vertices directly attached to $v$. 
A {\bf discrete m-sphere} is a discrete $q$-manifold which 
has the property that removing some vertex renders it contractible. 
Inductively, a graph is called {\bf contractible}, if there exists a vertex 
$v \in V$ such that both $S(v)$ and $S \setminus v$ are contractible.
The $1$-point graph $K_1=1$ is contractible. The empty graph is the $(-1)$-sphere. 

\paragraph{}
Let $G_k$ denote the set of $K_{k+1}$ subgraphs. They are also called 
$k$-simplices and $f_k=|G_k|$ is their cardinality. 
We have $G_0=V,G_1=E$. The {\bf Euler characteristic} of $G$ was
then defined by Schl\"afli as
$\chi(G) = \sum_{k=0}^q (-1)^k f_k = f_0-f_1+f_2-f_3+ \cdots + (-1)^q f_q$.
For $2$-manifolds it is $\chi(G)=f_0-f_1+f_2=V-E+F$.

\paragraph{}
A graph without edges is a $0$-manifold. A $0$-manifold is a $0$-sphere, if $V=2,E=0$
(removing a vertex produces $K_1$ which is contractible by definition). 
Every connected $1$-manifold is a $1$-sphere, a circular graph $C_n$ with $n \geq 4$. 
Every finite 2-manifold is either a $2$-sphere $S^2$ 
or a connected sum of tori or projective planes: 
$M=S^2, M=\mathbb{T}^2 \# \cdots \# \mathbb{T}^2$ or 
$M=\mathbb{P}^2 \# \cdots \# \mathbb{P}^2$.
A $2$-sphere is characterized as the only 2-manifold of Euler characteristic 2. 
The $16$ cell and the $600$ cells are examples of 3-spheres. The join of two
1-spheres is a 3-sphere. The join of a k-sphere with a m-sphere is 
a (k+m+1)-sphere. The join of $G$ with the 0-sphere is called {\bf suspension}. 

\paragraph{}
Euler's formula $\chi(G) = V-E+F=2$ for $2$-spheres $G$ generalizes to higher dimension.
The $0$-sphere has $\chi(M)=V=2$, every $1$-sphere has $\chi(M)=V-E=0$. 
Every $2$-sphere has $\chi(M)=V-E+F=2$. This pattern continues:
If $M$ is a $q$-sphere, then $\chi(M) = 1+(-1)^q$. This is the {\bf general Euler gem
formula}. The proof goes using induction with respect to dimension $q$: for $q=0$, we have $\chi(M)=2$.
The induction assumption is that all $(q-1)$-spheres $S$ satisfy $\chi(S) = 1+(-1)^{q-1}$. 
Pick a vertex $v$. As the unit sphere $S(v)$ is a $(q-1)$-sphere and 
$S(v) = B(v) \cap G \setminus v$, where both the unit ball $B(v)$ and $G \setminus v$
are contractible with Euler characteristic $1$, we have, using the induction assumption,
$\chi(M) = \chi(G \setminus v) + \chi(B(v)) - \chi(G \setminus v \cap B(v)) 
         = 2 - (1-(-1)^{q-1}) = 1+(-1)^q$. 

\subsection{Level Surfaces}

\paragraph{}
In the continuum, manifolds can be constructed as level surfaces
of functions. An example is  $x^2+y^2=1, x^2+y^2+z^2=9$, which is a disconnected 
$1$-dimensional variety consisting of 2 circles. This can be emulated also 
in the discrete. Take an arbitrary function on vertices $V$ which takes values in 
$Z_k=\{0, \dots, k\}$. It defines a new graph $G_f$, 
where the vertices are the set of complete subgraphs on which $f$ attains all $k$ values. 
Connect two of these points by an edge, if one is contained in the other. 
The new graph $G_f$ is a sub-graph of the {\bf Barycentric refinement} of $M$:

\paragraph{}
If $G$ is a simplicial complex and $x \in G$, define the core 
$S^-(x) = \{ y \in G, y \neq x, y \subset x \}$ and the star 
$S^+(x) = \{ y \in G, y \neq x, x \subset y \}$. The unit sphere $S(x)$ is the boundary of
$S^+(x)$ in the Alexandrov topology generated by all $S^+(x)$. 

\begin{lemma}
a) If $G$ is a complex and $H$ is contractible, then $G*H$ is contractible. \\
b) The join of a $k$-sphere with a $m$-sphere is a $k+m+1$-sphere.  \\
c) If $H$ is a complex and $G,G*H$ are spheres, then $H$ is a sphere.  \\
d) For any simplicial complex $G$ and any $x \in G$, we have $S(x) = S^-(x) * S^+(x)$.  \\
e) If $G$ is a $q$-manifold, then its Barycentric refinement is a $q$-manifold.
\end{lemma}

\begin{proof}
a) Use induction with respect to the number of elements in $G*H$. It is clear if $G$
and $H$ have one element. For $v$ in $G$, we have $S_{G * H}(v) = S_G(v) * H$ 
and $G*H \setminus v = (G \setminus v)*H$. 
For $v$ in $H$, we have $S_{G * H}(v) = G*S_H(v)$ and $G*H \setminus v = G*(H \setminus v)$. 
Because all of them are smaller they are contractible by induction.  \\
b) Use induction with respect to $n=k+m$. The induction assumption is that
for $k=0,m=0$ the join is $C_4$ which is a 1-sphere. 
Assume the statement is true for all $k+m$. Take a $(k+1)$-sphere $G$ and
a $m$-sphere $H$. Given $v \in G$, we have $S_{G*H}(v) = S_G(v) * H$. The right 
hand side is a sphere because $S_G(v)$ is a $k$ sphere and $H$ is a $m$-sphere. 
The same argument shows that if $w \in H$, $S_{G*H}(w)=G * S_H(w)$ is a sphere. 
Since every unit sphere in $G*H$ is a sphere, we have a manifold. To show that
$G*H \setminus w$ is contractible with $w \in H$, use that $G*H \setminus w = G*(H \setminus w)$
and a).  The case $v \in G$ is analog.  \\
c) Use induction with respect to dimension. Since $G,S_H(x)*G$ are spheres 
also $S_H(x)$ is a sphere.
d) By definition, every $y \in S(x)$ is either strictly contained in $x$ or then contains 
$x$ strictly. Every element in $S^-(x)$ is contained in $x$ and so in $S^+(x)$ so that
we indeed have a join. \\
e) We have to show that every $S(x)$ in the Barycentric refinement 
is a $(q-1)$-sphere. Since $S(x)=S^-(x) * S^+(x)$
we only need to show that $S^-(x)$ and $S^+(x)$ are spheres. But $S^-(x)$ is the boundary 
sphere of $x$ and so is a $(k-1)$-sphere if $x$ is $k$-dimensional. 
Having that $^-(x)$ and $S(x)$ are both spheres, d) shows that $S^+(x)$ is a sphere. 
\end{proof}

\paragraph{}
We prove now Theorem 3):  \\

\begin{center} \fbox{\parbox{10cm}{
If $G$ is a $m$-manifold and $f: G \to \{0,1,\dots, k\}$ is an arbitrary function, then either 
$G_f$ is empty or then $G_f$ is a $(m-k)$-manifold. 
}} \end{center} 

\begin{proof}
Let $x$ be a $n$-simplex on which $f$ takes all values. This means $f(x)=\{0,1, \dots, k\}$. The
graph $S^-(x) = \{ y \subset x, y \neq x\}$ is a $(n-1)$-sphere
in the Barycentric refinement of $M$. 
The simplices in $S^-(x)$ on which $f$ still reaches $\{0,1,\dots, k\}$ is by 
induction a $(n-1-k)$-manifold. Since we are in a simplex, it has to be
a $(n-1-k)$-sphere. Every unit sphere $S(x)$ in the Barycentric refinement is a 
$(m-1)$-sphere as it is the join of $S^-(x)$ with $S^+(x)=\{ y, x \subset y, 
x \neq y \}$. We used the above lemma showing that join of two spheres is always a sphere.
The sphere $S^+_f(x)$ in $G_f$ is the same than $S^+(x)$ in $M$ because every
simplex $z$ in $M$ containing $x$ automatically has the property that $f(z)=Z_k$. 
So, the unit sphere $S(x)$ in $G_f$ is the join of a $(n-k-1)$-sphere and 
the $m-n$-sphere and so a $(m-k-1)$-sphere. Having shown that every unit
sphere in $G_f$ is a $(m-k-1)$-sphere, we see that $G_f$ is a $(m-k)$-manifold.
\end{proof}

\pagebreak 

\subsection{Gauss-Bonnet} 

\paragraph{}
Differential geometry is much easier in the discrete, especially in higher dimensions, 
where the Gauss-Bonnet-Chern theorem is already quite tough to prove. 
Define {\bf curvature} of a vertex $v$ in an arbitrary simplicial complex $G$ or graph as
$K(v) = \sum_{k=0}^m \frac{(-1)^{k} f_{k-1}(S(v))}{k+1}  
       = 1-\frac{f_0(S(v))}{2}-\frac{f_1(S(v)))}{3} + \dots $. 
In the case of a $2$-manifold this boils down to $1-f_0(S(v))/2+f_1(S(v))/3=1-d(v)/6$,
where $d(v)$ is the vertex degree. For odd-dimensional manifolds, the curvature $K(v)$
is constant zero. 

\paragraph{}
The general Gauss-Bonnet theorem is

\begin{thm} $\sum_{v \in V} K(v) = \chi(M)$ \end{thm} 

\begin{proof} 
With energies $\omega(x)=(-1)^{{\rm dim}(x)}$ attached
to each simplex $x$ in the graph,
the Euler characteristic is $\chi(M) = \sum_{x} \omega(x)$. Now distribute all these energies of a $k$-simplex $x$ 
equally to the $k+1$ vertices contained in $x$. As there are $f_{k-1}$ simplices in $S(v)$
which correspond do simplices containing $v$,
this adds $\frac{(-1)^{k} f_{k-1}(S(v))}{k+1}$ to each vertex $v$. Now just collect all up
at a vertex $v$ to get the curvature $K(v)$. The transactions of energies preserved the total
energy = Euler characteristic. 
\end{proof} 

\paragraph{}
It is the discrete analog of the Gauss-Bonnet-Chern theorem (see \cite{Cycon}).

\section{Three questions}

\subsection{Isospectral set}

\paragraph{}
For periodic Jacobi matrices, the isospectral set is a torus
where the dimension depends on the number of gaps. 
See  \cite{Cycon,SimonGeszesy}. It is natural to ask what the structure
is in the case where the operator is $L=D^2$ for a complex $G$. 

\paragraph{}
{\bf Question A}: What is the isospectral set obtained from deformations in Theorem~1?

\subsection{The universal measure}

\paragraph{}
Every measure $\mu$ can be Lebesgue
decomposed into $\mu = \mu_{pp} + \mu_{sc} + \mu_{ac}$. 

\paragraph{}
{\bf Question B}: What are the spectral types of the limiting density of states
$dk_m$? 

\paragraph{}
In dimension $m=1$, where $L_0$ and $L_1$ are isospectral, 
the limiting measure is absolutely continuous, 
the $\arcsin$-distribution on $[0,4]$ with CDM
$(2/\pi) \arcsin(\sqrt{x}/2)$.

\subsection{Sub-manifold statistics}

\paragraph{} 
For a $m$-manifold $G$ and $k<m$, the set of all functions $f: V(G) \to \{0,\dots, k\}$
is finite. Put the uniform measure on it.
Let $b_f$ denote the Betti vector $(b_0,b_1, \dots,b_q)$ of $G_f$ with 
the assumption $b_f(\emptyset) = (0,0,\dots,0)$. The expectation 
${\rm E}[b_f(G)]$ is a vector $({\rm E}[b_0],{\rm E}[b_1], \dots,{\rm E}[b_q])$.

\paragraph{}
{\bf Question C}: Are the topological features of $G$ that relate to 
${\rm E}[b_f(G)]$? 

\pagebreak

\section{Computer algebra}

\subsection{Code for Isospectral Deformation}

We compute the deformation using the QR decomposition. 

\begin{tiny} 
\lstset{language=Mathematica} \lstset{frameround=fttt}
\begin{lstlisting}[frame=single]
Generate[A_]:=If[A=={},{},Sort[Delete[Union[Sort[Flatten[Map[Subsets,A],1]]],1]]];
Whitney[s_]:=Union[Sort[Map[Sort,Generate[FindClique[s,Infinity,All]]]]]; L=Length;
sig[x_]:=Signature[x];   nu[A_]:=If[A=={},0,L[A]-MatrixRank[A]];  omega[x_]:=(-1)^(L[x]-1);
F[G_]:=Module[{l=Map[L,G]},If[G=={},{},Table[Sum[If[l[[j]]==k,1,0],{j,L[l]}],{k,Max[l]}]]];
sig[x_,y_]:=If[SubsetQ[x,y]&&(L[x]==L[y]+1),sig[Prepend[y,Complement[x,y][[1]]]]*sig[x],0];
Dirac[G_]:=Module[{f=F[G],b,d,n=L[G]},b=Prepend[Table[Sum[f[[l]],{l,k}],{k,L[f]}],0];
   d=Table[sig[G[[i]],G[[j]]],{i,n},{j,n}];{d+Transpose[d],b}];
Hodge[G_]:=Module[{Q,b,H},  {Q,b}=Dirac[G];  H=Q.Q;
   Table[Table[H[[b[[k]]+i,b[[k]]+j]],{i,b[[k+1]]-b[[k]]},{j,b[[k+1]]-b[[k]]}],{k,L[b]-1}]];
Betti[s_]:=Module[{G},If[GraphQ[s],G=Whitney[s],G=s];Map[nu,Hodge[G]]];
Fvector[A_]:=Delete[BinCounts[Map[L,A]],1];  Euler[A_]:=Sum[omega[A[[k]]],{k,L[A]}];
QR[A_]:=Module[{F,B,n,T},T=Transpose;F=T[A];n[x_]:=x/Sqrt[x.x];B={n[F[[1]]]};Do[v=F[[k]];
 u=v-Sum[(v.B[[j]])*B[[j]],{j,k-1}];B=Append[B,n[u]],{k,2,Length[F]}];{T[B],B.A}];
QRDeformation[B_,t_]:=Module[{Q,R,EE},{Q,R}=QR[MatrixExp[-t*B]]; Transpose[Q].B.Q];
s=RandomGraph[{8,20}]; G=Whitney[s];{B,b}=Simplify[Dirac[G]]; 
Manipulate[B1 = QRDeformation[B,t]; MatrixPlot[Chop[B1]], {t, 0, 1}]
\end{lstlisting}
\end{tiny}

\subsection{Code for Barycentric Refinement}

We compute the spectrum in the case $q=2$ and see the 4th refinement of $K_3$. 

\begin{tiny} 
\lstset{language=Mathematica} \lstset{frameround=fttt}
\begin{lstlisting}[frame=single]
Generate[A_]:=If[A=={},{},Sort[Delete[Union[Sort[Flatten[Map[Subsets,A],1]]],1]]];
Facets[s_]:=FindClique[s,Infinity,All];        Whitney[s_]:=Generate[Facets[s]]; L=Length; 
F[G_]:=Module[{l=Map[L,G]},If[G=={},{},Table[Sum[If[l[[j]]==k,1,0],{j,L[l]}],{k,Max[l]}]]];
Fvector[s_]:=F[Whitney[s]]; Euler[A_]:=Sum[-(-1)^L[A[[k]]],{k,L[A]}]; 
ToGraph[G_]:=UndirectedGraph[n=L[G];Graph[Range[n],
         Select[Flatten[Table[k->l,{k,n},{l,k+1,n}],1],(SubsetQ[G[[#[[2]]]],G[[#[[1]]]]])&]]];
Barycentric[s_]:=ToGraph[Whitney[s]];  Barycentric[s_,n_]:=Last[NestList[Barycentric,s,n]]; 
s=Barycentric[CompleteGraph[3],4];K=KirchhoffMatrix[s];n=L[K];l=Quiet[Sort[Eigenvalues[1.K]]]; 
UU[z_]:=Sum[If[Abs[l[[k]]]< 10^(-10),0,Log[Abs[l[[k]]-z]]],{k,L[l]}];
S1=ContourPlot[UU[x+I y],{x,-2,10},{y,-6,6},ContourStyle->Thickness[0.0001],PlotPoints->200];
S2=Plot[UU[x]/n, {x,-2, 10},PlotStyle->{Thickness[0.002],Purple}];
S3=Graphics[{Black,PointSize->0.003,Table[Point[{l[[k]], 0}],{k,L[l]}]}];
S4=Graphics[{Blue,Thickness[0.002],Line[Prepend[Table[{l[[k]],4*(k-1)/n},{k,L[l]}],{-2,0}]]}];
S=Show[{S1,S2,S3,S4},PlotRange->{{-2,10}, {-4,5}},AspectRatio->9/16]
\end{lstlisting}
\end{tiny}

\subsection{Code for Level sets in Manifolds}

Take a 4-sphere and compute manifolds in it and check the Gauss-Bonnet formula:

\begin{tiny}
\lstset{language=Mathematica} \lstset{frameround=fttt}
\begin{lstlisting}[frame=single]
Generate[A_]:=If[A=={},{},Sort[Delete[Union[Sort[Flatten[Map[Subsets,A],1]]],1]]];
Whitney[s_]:=Generate[FindClique[s,Infinity,All]]; w[x_]:=-(-1)^k;
R[G_,k_]:=Module[{},R[x_]:=x->RandomChoice[Range[k]]; Map[R,Union[Flatten[G]]]];
F[G_]:=Delete[BinCounts[Map[Length,G]],1]; Euler[G_]:=F[G].Table[w[k],{k,Length[F[G]]}];
Surface[G_,g_]:=Select[G,SubsetQ[#/.g,Union[Flatten[G] /. g]] &]; Clear[S]; 
S[s_,v_]:=VertexDelete[NeighborhoodGraph[s,v],v];     Sf[s_,v_]:=F[Whitney[S[s,v]]];
Curvature[s_,v_]:=Module[{f=Sf[s,v]},1+f.Table[(-1)^k/(k+1),{k,Length[f]}]];
Curvatures[s_]:=Module[{V=VertexList[s]},Table[Curvature[s,V[[k]]],{k,Length[V]}]];
J[G_,H_]:=Union[G,H+Max[G]+1,Map[Flatten,Map[Union,Flatten[Tuples[{G,H+Max[G]+1}],0]]]];
ToGraph[G_]:=UndirectedGraph[n=Length[G];Graph[Range[n],
  Select[Flatten[Table[k->l,{k,n},{l,k+1,n}],1],(SubsetQ[G[[#[[2]]]],G[[#[[1]]]]])&]]];
Barycentric[s_]:=ToGraph[Whitney[s]];
G=J[Whitney[Barycentric[CompleteGraph[{2,2,2}]]],Whitney[CycleGraph[7]]];  (* J=Join *)
g=R[G,3]; H=Surface[G,g];   (* A codimension 2 manifold in the 4-sphere G=Oct * C_7  *)
Print["EulerChi= ",Euler[H]]; Print["Fvector: ",F[H]];s=ToGraph[H];GraphPlot3D[s]
Print["Gauss-Bonnet Check: "]; Print[Total[Curvatures[s]]==Euler[H]];
\end{lstlisting}
\end{tiny}

\bibliographystyle{plain}

\end{document}